\documentclass{llncs}

\usepackage{amssymb}
\usepackage{amsmath}
\usepackage{color}

\usepackage{algorithm}
\usepackage{algorithmic}
\usepackage{graphicx}
\usepackage{subfigure}
\usepackage{enumitem}

\usepackage{comment}

\newcommand{\always}[2]{\mathbf{G}_{[#1,#2]}}

\begin{document}

\mainmatter

\title{Policy learning for time-bounded reachability in Continuous-Time Markov Decision Processes via doubly-stochastic gradient ascent}

\author{
Ezio Bartocci\inst{1}, Luca Bortolussi\inst{2,3,4}, Tom\v{a}\v{s} Br\'{a}zdil\inst{5},\\ Dimitrios Milios\inst{6}, Guido Sanguinetti\inst{6,7}
}

\institute{
Faculty of Informatics, Vienna University of Technology, Austria
\and
Dept. of Maths and Geosciences, University of Trieste, Italy
\and
CNR/ISTI, Pisa, Italy
\and
Modelling and Simulation Group, Saarland University, Germany
\and
Faculty of Informatics, Masaryk University, Czech Republic
\and
School of Informatics, University of Edinburgh, UK
\and
SynthSys, Centre for Synthetic and Systems Biology, University of Edinburgh, UK
}

\maketitle

\begin{abstract}
Continuous-time Markov decision processes are an important class of models in a wide range of applications, ranging from cyber-physical systems to synthetic biology. A central problem is how to devise a policy to control the system in order to maximise the probability of satisfying a set of temporal logic specifications. Here we present a novel approach based on statistical model checking and an unbiased estimation of a functional gradient in the space of possible policies. The statistical approach has several advantages over conventional approaches based on uniformisation, as it can also be applied  when the model is replaced by a black box, and does not suffer from state-space explosion. The  use of a stochastic gradient to guide our search considerably improves the efficiency of learning policies. We demonstrate the method on a proof-of-principle non-linear population model, showing strong performance in a non-trivial task.  

\end{abstract}

\section{Introduction}
\label{sec:intro}
% !TEX root =  CTMDP_gradDescent.tex

Continuous-time Markov Decision Processes (CTMDPs)~\cite{Baier2005} are a very powerful 
mathematical framework to solve control and dependability problems in real-time 
systems featuring both probabilistic and nondeterministic behaviours. 
Examples include applications such as the control of epidemic processes~\cite{lefevre81,Xianping2006}, 
 power management~\cite{Qui2001}, queueing systems~\cite{Sennott1998}  
and cyber-physical systems~\cite{Ayala2012}. A CTMDP extends a continuous-time Markov chain (CTMC) by 
introducing a decision maker (also called \emph{scheduler}) that can perform 
actions with an associated cost or reward.
CTMDPs are particularly useful modelling tools to address important problems 
such as \emph{model checking}~\cite{Baier2003}  and \emph{planning}. 

\emph{Model checking} aims to verify if a CTMDP satisfies a desired  
requirement for a given class of schedulers or  for all possible schedulers.
 The requirement of interest is usually expressed in terms of the \emph{min}/\emph{max}
 probability for a CTMDP to satisfy the temporal logic property~\cite{Baier2003} of interest. 
 In particular, the main target of the current quantitative model checking techniques for CTMDPs
is the \emph{time-bounded reachability}~\cite{Baier2005,Neuhaeusser2010,Rabe2011,Rabe2013,Yuliya2015}, a property that requires a CTMDP
to reach a particular set of states within a time bound.

\emph{Planning} or \emph{scheduling} is an orthogonal problem w.r.t. model checking. It consists in devising  
the optimal sequence of actions (or \emph{policy}) to control the system in order to 
maximise the probability to satisfy a temporal logic specification such as the 
aforementioned time-bounded reachability. In the case of CTMDP the optimal 
scheduling can be either \emph{timed} or \emph{untimed} depending on whether or not 
the scheduler is aware of the passing of time. Timed optimal scheduling can be further 
classified in \emph{late} or \emph{early} depending on whether the decision of choosing an action 
 can change while the time passes in a state or it remains unchanged. 

In this paper we present a novel statistical approach to compute lower bounds on the 
 maximum reachability probability of a CTMDP. Our method uses a basis-function regression approach to compactly encode schedulers and effectively search for an 
 optimal one. We consider here \emph{randomised time-dependent early schedulers}, and  focus on population models, where the state  space of the CTMDP is represented by a set of integer-valued variables counting how 
 many entities of each kind are in the system. This is a large class of models: queueing 
 and performance models~\cite{Yuliya2015},  epidemic scenarios, biological systems 
 are all members of this class. Population models, despite being so common, suffer 
 severely from state space explosion, with the number of states growing exponentially with 
 the number of variables.  This reflects on the size of the schedulers: in principle, we would   need to store a function of time for each state of the CTMDP,  which is unfeasible. This paper contains two main novel insights. First, we leverage the structure of the state space, which can be embedded as 
 a discrete grid in real space, to obtain a continuous relaxation of the problem and  consider schedulers defined on such a continuous space. 
 The advantage now is that we can treat time and space uniformly, representing schedulers 
 as continuous functions. This opens up the use of machine learning methods 
 to represent continuous functions as combinations of basis functions, and allows us to define 
 the optimisation problem as a search in such a continuous function space. The second main contribution of the work is  to set up an efficient stochastic gradient ascent search algorithm, which considerably speeds up the search in the space of functions. This is based on a novel algorithm using Gaussian Processes (GPs) and statistical model checking to
  sample in an unbiased manner the gradient of the functional associating a 
 reachability probability with a randomized scheduler.  This method allows us to effectively learn schedulers that maximise (locally) the reachability 
 probability. 
 
 \paragraph{Organisation of the paper.}  In Section~\ref{related_work} we present the related 
 work and in Section~\ref{sec:preliminaries} we provide the 
 necessary formal background on CTMDPs.
  In Section~\ref{sec:gradientDescent} we present our algorithm to learn optimal policies using stochastic 
 functional gradient ascent techniques. In Section~\ref{sec:example} we demonstrate our algorithm 
 on an epidemiology case study. Finally, we draw our conclusion in Section~\ref{sec:conclusions}.
 
 \section{Related work}
 \label{related_work}
 
Symbolic model checking algorithms for discrete-time Markov 
decision processes have been intensively investigated in~\cite{baier98,alfaro1995} and implemented 
in popular tools such as PRISM~\cite{kwiatkowska_prism_2011}. 
In the area of CTMDPs, the problem of time optimal planning 
has been first considered from a theoretical point of view in~\cite{miller68}. In the last decade  there has been a great effort
on developing practical model checking techniques for CTMDPs~\cite{Baier2005,Neuhaeusser2010,Rabe2011,Rabe2013,Yuliya2015} 
(i.e., based on  uniformization~\cite{Baier2005}) with the introduction of efficient approximation algorithms 
that provide also formal error bounds. Generally, all these techniques rely on the a-priori 
 knowledge of the CTMDP model under investigation and they suffer the state-explosion problem.

In this light, methods based on statistical model checking are particularly attractive, even though 
they may suffer when the property to be verified is a rare-event.
In \cite{HenriquesMZPC12} the authors presented a statistical model checking algorithm for the discrete-time case; 
their approach was however based on random search combined with a greedy selection 
criterion, which is difficult to analyse in terms of convergence properties, and may 
be practically difficult to tune.  The availability of an unbiased estimate of the (functional) gradient allows us to improve on 
the efficiency, and to leverage a rich theory on the convergence of stochastic gradient ascent algorithms. 
 Our approach relies on using Gaussian Processes (GPs), a probability distribution over the space of functions 
which universally approximates continuous functions. This ability of GPs to provide efficient approximations to 
intractable functions has been recently exploited in 
a formal modelling context in a number of publications \cite{bortolussi:smoothed16,Bartocci2015,lucaQEST13}.

Our work is closely related to research in the area of machine learning, where much research has gone on defining good local search methods to learn effective randomised 
schedulers, for different criteria like time bounded reward, time unbounded 
discounted reward, receding horizon.  These approaches combine simulation 
with efficient exploration schemes, like gradient ascent~\cite{rosenstein_robot_2001,bartlett_experiments_2011}, 
path integral policy improvement~\cite{stulp_path_2012}, or the cross 
entropy method~\cite{mannor_cross_2003},  see~\cite{stulp_policy_2012}  for a survey. Our approach differs in two main directions: firstly, we are interested in complex rewards associated with trajectories of the system, i.e. reachability probabilities. Secondly, we work directly in continuous time, which prevents the use of simple finite-dimensional gradient ascent methods. In particular, the GP-based method of defining a stochastic gradient ascent algorithm is novel, to the best of our knowledge.

%\section{Related Work}
%\label{sec:related}
%\input{sec_related}

\section{Preliminaries}
%\guido{I think this should also include a formal definition of CTMC and particularly the introduction of the concept of reward and expected reward}
\label{sec:preliminaries}

%PRELIMS - TOMAS

\begin{definition}
A continuous-time Markov decision process (CTMDP) is a tuple $\mathcal{M} = (S,\mathcal{A},R,s_0)$, where $S$ is a finite set of {\em states}, $\mathcal{A}$ is a finite set of {\em actions},  $R:S \times\mathcal{A}\times S\rightarrow \mathbb{R}_{\geq 0}$ is the {\em rate function}, and $s_0\in S$ is the initial state. 
\end{definition}
\noindent
An action $a\in \mathcal{A}$ is {\em enabled} in a state $s\in S$ if there is a state $s'\in S$ such that $R(s,a,s')>0$. We call $\mathcal{A}(s)$ the set of enabled actions in $s$. A {\em continuous-time Markov chain (CTMC)} is a CTMDP where every $\mathcal{A}(s)$ is a singleton.

We define $E(s,a) = \sum_{s'} R(s,a,s')$ the {\em exit} rate from a state $s$ when an~action $a$ is chosen. We also let $P(s,a,s') = R(s,a,s')/ E(s,a)$ be the probability of jumping from $s$ to $s'$ if $a$ is selected. 

Intuitively, a run of CTMDP starts in a state $s_0$ and proceeds as follows: Assume that the CTMDP is currently in a state $s_i$. First, an action $a_i$ is selected, then the CTMDP waits for a delay $t_i$ randomly chosen according to an exponential distribution with the exit rate $E(s_i,a_i)$, and then a next state $s_{i+1}$ is chosen randomly with the probability $P(s_i,a_i,s_{i+1})$. This produces a run $s_0 a_0 t_0 s_1 a_1 t_1 \cdots$.

In order to obtain a complete semantics, we need to specify how the actions are selected in every step.
Obviously, in CTMC, only a single action is enabled in each state. In CTMDP, actions need to be chosen by a scheduler defined as~follows.
\begin{definition}\label{def:scheduler}
An {\em (early timed) scheduler} is a function $\sigma:\mathbb{R}_{\geq 0}\times S\times \mathcal{A}\rightarrow [0,1]$ which to every $t\in \mathbb{R}_{\geq 0}$, $s\in S$ and $a\in \mathcal{A}$ assigns a probability measure $\sigma(t,s,a)$ that the~action $a$ is chosen in $s$ at time $t$. 
\end{definition}
%\guido{I changed probability to probability measure since we also have a continuous variable $t$}
A scheduler $\sigma$ is {\em deterministic} if for every $t\in \mathbb{R}_{\geq 0}$, $s\in S$ and $a\in \mathcal{A}$ we have that $\sigma(t,s,a)\in \{0,1\}$. We denote by $\Sigma$ and $\Sigma_{D}$ the sets of all schedulers and all deterministic schedulers, respectively.

\begin{remark}
An early scheduler has the following property:  whenever an execution of the CTMDP enters into a state $s$ at time $t$, the scheduler chooses an action and commits to it. It cannot be changed while the system remains in state $s$, in contrast with late schedulers, that can change action while in a state. 
\end{remark}
Once a scheduler $\sigma$ and an initial state $s$ is fixed, we obtain the unique probability measure $\mathbb{P}_{\sigma}^{\mathcal{M},s}$ over the space of all runs initiated in $s$ using standard definitions~\cite{Neuhausser2010}.

\paragraph{Time-Bounded Reachability.} Let $G\subset S$ be a set of goal states and let $I=[t_1,t_2]\subseteq [0,\infty)$ be a closed interval.
Denote by $\mathbb{P}^{\mathcal{M},s}_{\sigma}(\diamond_{I} G)$ the probability that $G$ is reached from $s$ within the time interval $I$ using the scheduler $\sigma$. Our goal is to maximize $\mathbb{P}^{\mathcal{M},s}_{\sigma}(\diamond_{I} G)$, i.e. compute a scheduler $\sigma^*$ satisfying
\[
\mathbb{P}^{\mathcal{M},s}_{\sigma^*}(\diamond_{I} G)\quad = \quad \sup_{\sigma\in\Sigma} \mathbb{P}^{\mathcal{M},s}_{\sigma}(\diamond_{I} G)
\]
We say that such a scheduler $\sigma^*$ is {\em optimal}.
\begin{proposition}[\cite{Neuhausser2010}]\label{prop:Neuhausser}
There always exists an optimal scheduler.
\end{proposition}
%\tomas{Do we want a shorthand $\mathit{val}(s)$ for $\mathbb{P}^{\mathcal{M},s}_{\sigma}(\diamond_{\leq T} G \vert \sigma)$?}
When dealing with  time-bounded reachability, we may safely assume that schedulers are defined only on the interval $[0,T]$, i.e., on a compact set.
An equivalent problem is to maximise a time-bounded safety property $\square_I G$, requiring the CTMDP to remain in a region $G$ during the time-interval $I$. In this case, we have that $\mathbb{P}^{\mathcal{M},s}_{\sigma^*}(\square_{I} G) = \mathbb{P}^{\mathcal{M},s}_{\sigma^*}(\neg\diamond_{I} S\setminus G) = \inf_{\sigma\in\Sigma} \mathbb{P}^{\mathcal{M},s}_{\sigma}(\diamond_{I} S\setminus  G)$.

\paragraph{Population CTMDPs.} In this work, we will consider CTMDPs modelled in a special way, reminiscent of population processes which are very common in performance modelling, epidemiology, systems biology. The basic idea is that we will have  populations of  agents, belonging to one or more classes,  that can interact together and thus evolve in time. Individual agents are typically indistinguishable, hence the state of the system can be described by a set of  variables counting the amount of agents of each kind in the system. A non-deterministic action in this context typically represents an action of a global controller, enforcing a policy controlling the system, or effects on the environment. 

More formally, we will describe a Population CTMDP (PCTMDP), extending population processes \cite{tutorial,wolf11}, as a tuple $(\vec{X},\mathcal{T},\mathcal{A},\vec{s}_0)$, where:
\begin{itemize}
\item $\vec{X} = X_1,\ldots,X_n$ is a vector of population variables, $X_i\in\mathbb{N}$, which we assume take values on $S = \mathbb{N}^n\cap E$, where $E$ is a compact subset of $\mathbb{R}^n$ (hence $S$ is finite);
\item $\vec{s_0}\in S$ is the initial state;
\item $\tau\in \mathcal{T}$ is the set of transitions, of the form $(a,\vec{v},f(\vec{X}))$, where $a$ is an action from the set $\mathcal{A}$, $\vec{v}$ is an update vector, specifying that the state after the execution of a transition in state $\vec{s}$ is $\vec{s}+\vec{v}$, and $f(\vec{X})$ is the state-dependent rate function.
\end{itemize}
The idea of this model is that in each state an action $a$ is chosen, and then the model evolves by a race condition between  transitions guarded by the action $a$. 
If a transition is  enabled by all possible actions, we can either specify a copy of it guarded by each model action $a$, or use the notation  $(*,\vec{v},f(\vec{X}))$. 
%\tomas{I am not sure that I understand this ... Does it mean that $(*,\vec{v},f(\vec{X}))$ is enabled everywhere?}
The CTMDP $\mathcal{M} = (S,\mathcal{A},R)$ associated with a PCTMDP $(\vec{X},\mathcal{T},\mathcal{A},\vec{x}_0)$ is defined by specifying the state space  $S = \mathbb{N}^n\cap E$ and the rate function $R$ as 
\[ R(\vec{s},a,\vec{s'}) =  \sum\{ f_\tau(\vec{s})~|~\tau=(a,\vec{v},f(\vec{s}))\wedge \vec{s'} = \vec{s}+\vec{v}\}.\]
It is easy to observe, modulo the introduction of enough variables and actions, that the expressive power of PCTMDPs is the same as that of CTMDPs introduced earlier.

\section{Learning optimal policies via stochastic functional gradient ascent}
\label{sec:gradientDescent}
In this section we give a variational formulation of the control problem of determining the optimal scheduler for a CTMDP. 
We show how to approximate statistically in an unbiased way the functional gradient of the time-bounded reachability probability, and give a convergent algorithm to achieve this. 

% This formulation is equivalent to the classical Hamilton-Jacobi-Bellman formulation of continuous time control problems when solved through fixed point equations, however it can also be used as the starting point for a direct optimisation. 
% \guido{Some of this could go at the end of the intro instead}

\subsection{Reachability probability as a functional}

As defined in Section \ref{sec:preliminaries}, a scheduler is a way of resolving non-determinism by associating a (time-dependent) probability to each action/ state pair. 
We will realise a scheduler as a vector $\mathbf{f}$ of  functions $f_{\alpha}: E \times [0,T] \rightarrow \mathbb{R}$, one for each action $\alpha \in \mathcal{A}$, where $E$ is the compact subset of $\mathbb{R}^n$ used to define $S$ for the PCTMDP formalism. The corresponding probability of an action $\alpha$ at a state $\vec{X}$ can be retrieved using the soft-max (logistic) transform as follows:
\begin{equation}
p_{\vec{X}}(\alpha \mid t) \equiv \sigma(t,\alpha,\vec{X}) = \frac{\exp(f_{\alpha}(\vec{X}, t))}{\sum_{\alpha' \in \mathcal{A}} \exp(f_{\alpha'}(\vec{X}, t))}, \qquad \vec{X} \in S, t\in [0,T]
\end{equation}
Given a scheduler $\sigma$, a CTMDP is reduced  to a CTMC $\mathcal{M}_{\sigma}$, and the problem of estimating the probability of a reachability property $\phi = \diamond_I G$ can be reduced to the computation of a transient probability for $\mathcal{M}_{\sigma}$ by standard techniques \cite{Baier2003}.
%
%
%is equivalent to the probabilistic model checking problem of computing the satisfaction probability of the corresponding formula $\phi$ for $\mathcal{M}_{\sigma}$.
The satisfaction probability can be therefore viewed as a {\it functional} 
\[
Q\colon\mathcal{F}\rightarrow\mathbb{R}
\]
where $\mathcal{F}$ is the set of all possible scheduler functions.
The functional is defined explicitly as follows: consider a sample trajectory $\{s, a, t\}_n \equiv s_0\xrightarrow{\alpha_0,t_0}s_1\xrightarrow{\alpha_1,t_1}\ldots s_n\xrightarrow{\alpha_n,t_n}s_{n+1}$ from the CTMC $\mathcal{M}_{\sigma}$ obtained from the CTMDP by selecting a scheduler. Let $\phi= \diamond_I G$, $I=[t_1,t_2]$ be a reachability property, and denote by $ \{s, a, t\}_n \models \phi$ the fact that the trajectory reaches $G$ within the specified time bound. We can encode it in the following indicator function:
\begin{equation}
I_{\phi}(\{s, a, t\}_n) = 
\begin{cases}
1, \quad \{s, a, t\}_n \models \phi \\
0, \quad \mathrm{otherwise}.
\end{cases}
\end{equation}
Then the expected reachability value associated with the scheduler $\sigma$, represented by the vector of functions $\mathbf{f}=\{ f_\alpha\}_{\alpha\in\mathcal{A}}$, is defined as follows:
\begin{equation}
Q\left[\mathbf{f}(\vec{X}, t)\right]=E_{\mathcal{M}_{\sigma}}\left[I_{\phi}(\{s, a, t\}_n)\right],
\label{reward}
\end{equation}
where expectation is taken with respect to the distribution on trajectories of $\mathcal{M}_{\sigma}$.
Notice that in general it is computationally very hard to analytically compute the r.h.s.\ in the above equation, as it amounts to transient analysis for a time-inhomogeneous CTMC; we therefore need to resort to statistical model checking methods \cite{SB:Zuliani:2009:StatMC,MC:YounesSimmons:INFCOMP2006:statMC} to approximate in a Monte Carlo way the expectation in equation \eqref{reward}.

To formulate the continuous time control problem of determining the optimal scheduler, we need to define the concept of functional derivative.

\begin{definition}
Let $Q\colon\mathcal{F}\rightarrow\mathbb{R}$ be a functional defined on a space of functions $\mathcal{F}$. 
The {\it functional derivative} of $Q$ at $f\in\mathcal{F}$ along a function $g\in\mathcal{F}$, denoted by $\frac{\delta Q}{\delta f}$, is defined by 
\begin{equation}
\int \frac{\delta Q}{\delta f}(\vec{X}, t) \, g(\vec{X}, t) \, ds dt = \lim_{\epsilon \rightarrow 0} \frac{Q[f(\vec{X}, t) + \epsilon g(\vec{X}, t)] - Q[f(\vec{X}, t)]}{\epsilon}
\label{eq:directionalDerivativeApprox}
\end{equation}
whenever the limit on the r.h.s. exists.  % along a function 
\end{definition}
Notice that if we restrict ourselves to piecewise constant functions on a grid, the definition above returns the standard definition of gradient of a finite-dimensional function.
We can now give a variational definition of optimal scheduler
\begin{lemma}
An optimal scheduler $\sigma$ is associated with a function $f$ such that 
\begin{equation}
max_{g\in\mathcal{F}} \left\| \int \frac{\delta Q}{\delta f}(\vec{X}, t) \, g(\vec{X}, t) \, ds dt \right\|_2=0\label{varFormulation}
\end{equation}
where $\Vert\cdot\Vert_2$ denotes the $L^2$ norm on functions.
\end{lemma}
% Under mild conditions, the variational problem above has a unique solution, and its solution via fixed-point equations is equivalent to the Hamilton-Jacobi-Bellman (HJB) equations. The HJB equations unfortunately heavily suffer from the state space explosion problem, and as such often are not a practical algorithmic approach. 
The variational formulation above allows us to attack the problem via direct optimisation through a gradient ascent algorithm, as we will see below.

\subsection{Stochastic Estimation of the Functional Gradient}

It is well-known that a gradient ascent approach is guaranteed to find the global optimum of a convex objective function.
Gradient ascent starts from an initial solution which is updated iteratively towards the direction that induces the steepest change in the objective function; that direction is given by the gradient of the function.
For a functional $Q[f]$ the concept of gradient is captured by the functional derivative $\frac{\delta Q}{\delta f}$, which is a function of $\vec{X}, t$ that dictates the rate of change of the functional $Q$ when $f$ is perturbed at the point $(\vec{X}, t)$.
In the case of functional optimisation, the gradient ascent update will have the form:
\begin{equation}
f' = f + \gamma \frac{\delta Q}{\delta f}
\end{equation}
where $\gamma$ is the learning rate which controls the effect of each update, and $\frac{\delta Q}{\delta f}$ is the functional derivative of $Q$. Unfortunately, an analytic expression for the functional derivative of the functional defined in \eqref{reward} is usually not available.

We can however obtain an unbiased estimate of the functional derivative by using the infinite-dimensional generalisation of this simple lemma
\begin{lemma}\label{finDimGrad}
Let $q\colon\mathbb{R}^n\rightarrow\mathbb{R}$ be a smooth function, and let $\nabla q(\mathbf{v})$ be its gradient at a point $\mathbf{v}$. 
Let $\mathbf{w}$ be a random vector from an isotropic, zero mean distribution $p(\mathbf{w})$. 
For $\epsilon\ll 1$, define  \begin{equation}
\hat{\mathbf{w}}= 
\begin{cases}
\mathbf{w}, \quad  \mathrm{if\ } q(\mathbf{v}+\epsilon\mathbf{w}) - q(\mathbf{v}) > 0 \\
-\mathbf{w}, \quad \mathrm{otherwise}.
\end{cases}\label{flipper}
\end{equation}
Then\[
E_p\left[\epsilon\hat{\mathbf{w}}\right]\propto\nabla q(\mathbf{v})+O(\epsilon^2).\]
\end{lemma}
\begin{proof}
The tangent space of $\mathbb{R}^n$ at the point $\mathbf{v}$ is naturally decomposed in the orthogonal direct sum of a subspace of dimension 1 parallel to the gradient, and a subspace of dimension $n-1$ tangent to the level surfaces of the function $q$. 
For small $\epsilon$, any change in the value of the function $q$ will be due to movement in the gradient direction. 
As the distribution $p$ is isotropic, every direction is equally likely in $\mathbf{w}$; however, the flipping operation in the definition of $\hat{\mathbf{w}}$ in \eqref{flipper} ensures that the component of $\hat{\mathbf{w}}$ along the gradient $\nabla q(\mathbf{v})$ is always positive, while it does not affect the orthogonal components. 
Therefore, in expectation, $\hat{\mathbf{w}}$ returns the direction of the functional gradient.
\end{proof}

\subsection{Scheduler representation in terms of basis functions}

In order to obtain an unbiased estimate of a functional gradient, we need to define a zero-mean isotropic distribution on a suitable space of functions. To do so, we introduce the concept of Gaussian Process, a generalisation of the multivariate Gaussian distribution to infinite dimensional spaces of functions (see, e.g. \cite{Rasmussen2006}).

\begin{definition} A Gaussian Process (GP) over an input space $\mathcal{X}$ is an infinite-dimensional family of real-valued random variables indexed by $x\in\mathcal{X}$ such that, for every finite subset $X\subset\mathcal{X}$, the finite dimensional marginal obtained by restricting the GP to $X$ follows a multi-variate normal distribution.
\end{definition}
Thus, a GP can be thought as a distribution over functions $f\colon\mathcal{X}\rightarrow\mathbb{R}$ such that, whenever the function is evaluated at a finite number of points, the resulting random vector is normally distributed. In the following, we will only consider $\mathcal{X}=\mathbb{R}^d$ for some integer $d$.

Just as the Gaussian distribution is characterised by two parameters, a GP is characterised by two functions, the {\it mean} and {\it covariance} function. The mean function plays a relatively minor role, as one can always add a deterministic mean function, without loss of generality; in our case, since we are interested in obtaining small perturbations, we will set it to zero. The covariance function, which captures the correlations between function values at different inputs, instead plays a vital role, as it defines the type of functions which can be sampled from a GP. We will use the {\it Radial Basis Function} (RBF) covariance, defined as follows:
\begin{equation}
\mathrm{cov}(f(x_1),f(x_2))=k(x_1,x_2)=\alpha^2\exp\left[-\frac{\Vert x_1-x_2\Vert^2}{\lambda^2}\right].\label{rbfCov}
\end{equation}
where $\alpha$ and $\lambda$ are the amplitude and length-scale parameters of the covariance function.
To gain insight into the geometry of the space of functions associated with a GP with RBF covariance, we report without proof the following lemma (see e.g. Rasmussen \& Williams, Ch 4.2.1 \cite{Rasmussen2006}).
\begin{lemma}\label{denseLemma}
Let $\mathcal{F}_N$ be the space of random functions $f=\sum_{j=1}^N w_j\phi_j(x)$  generated by taking linear combinations of basis functions  $\phi_j(x)=\exp\left[-\frac{\Vert x-\mu_j\Vert^2}{\lambda^2}\right]$, with $\mu_j\in\mathbb{R}$ and independent Gaussian coefficients $w_j\sim\mathcal{N}(0,\alpha^2/N)$.
The sample space of a GP with RBF covariance defined by \eqref{rbfCov} is the infinite union of the the spaces $\mathcal{F}_N$.
\end{lemma}

We refer to the basis functions entering in the constructive definition of GPs given in Lemma \ref{denseLemma} as {\it kernel functions}. Two immediate consequences of the previous Lemma are important for us:\begin{itemize}
\item{A GP with RBF covariance defines an {\it isotropic} distribution in its sample space (this follows immediately from the i.i.d.\ definition of the weights in Lemma \ref{denseLemma})};
\item{The sample space of a GP with RBF covariance is a dense subset of the space of all continuous functions (see also \cite{bortolussi:smoothed16} and references therein).}
\end{itemize}

GPs therefore provide us with a convenient way of extending the procedure described in Lemma \ref{finDimGrad} to the infinite dimensional setting.
In particular, Lemma \ref{denseLemma} implies that any scheduler function $f \in \mathcal{F}$ that is a sample from a GP (with RBF covariance) can be approximated to arbitrary accuracy in terms of basis functions as follows:
\begin{equation}
f(\vec{X}, t) = \sum_{j=1}^N w_j \exp\left[-0.5 ([\vec{X},t]^{\top}-\mu_j)^{\top} \Lambda^{-1} ([\vec{X},t]^{\top}-\mu_j) \right]
\end{equation}
where $\mu_j \in \mathbb{R}^n \times [0,T]$ is the centre of a Gaussian kernel function, $\Lambda$ is a diagonal matrix that contains $n+1$ squared length-scale parameters of the kernel functions, and $n$ is the dimensionality of the state-space.
This formulation allows describing functions (aka points in an infinitely dimensional Hilbert space) as points in the finite vector space spanned by the weights $\mathbf{w}$.
Note that the proposed basis function representation implies relaxation of the population variables to the continuous domain, though in  practice we are only interested in evaluating $f(\vec{X}, t)$ for integer-valued $\vec{X}$. 

The advantage of the kernel representation is that we do not need to account for all states $\vec{X} \in S$, but only for $N$ Gaussian kernels with centres $\mu_j$ for $1 \leq j \leq N$.
Therefore, the value of the scheduler at a particular state $\vec{X}$ will be determined as a linear combination of the kernel functions, with proximal kernels contributing more due to the exponential decay of the kernel functions.
This method offers a compact representation of the scheduler, and essentially does not suffer from state-space explosion, as we treat states as continuous. Moreover, we do not lose accuracy, as every function on $S$ can be extended to a continuous function on $E$ by interpolation.
On the practical side, we consider that the kernel functions are spread evenly across the joint space (state space \& time), and the length-scale for each dimension is considered to be equal to the distance of two successive kernels.\footnote{Kernel functions typically also have an amplitude parameter, which we consider to be equal to 1.}

\subsection{A Stochastic Gradient Ascent Algorithm}

Given a scheduler $\sigma$, we first evaluate the reachability probability via statistical model checking. 
We then perturb the corresponding functions $f_{\alpha}$ by adding a draw from a zero-mean GP with marginal variance scaled by $\epsilon\ll 1$, and evaluate again by statistical model checking the probability of the perturbed scheduler. If this is increased, we take a step in the perturbed direction, otherwise we take a step in the opposite direction. Notice that this procedure can be repeated for multiple independent perturbation functions to obtain a more robust estimate. 
The whole procedure is described in Algorithm \ref{alg:estimateGradient}, which produces an estimate for the gradient of the functional $Q$ at a vector $\mathbf{f}$ of functions $f_{\alpha}$ by considering the average of $k$ random directions.
\begin{algorithm}[ht!]
\caption{Estimate the functional gradient of $Q[\mathbf{f}]$}
\label{alg:estimateGradient}
\begin{algorithmic}
\REQUIRE Vector $\mathbf{f}$ of functions $f_{\alpha}$, scaling factor $\epsilon$, batch size $k$
\ENSURE An estimate of the functional derivative (gradient) $\nabla Q \equiv \frac{\delta Q}{\delta \mathbf{f}}$
	\STATE Set gradient $\nabla Q = 0$
	\STATE Evaluate $Q[\mathbf{f}]$ via statistical model checking
	\FOR{$i=1$ to $k$}
		\STATE Consider random direction $\mathbf{g}$ such that $\forall \alpha \in \mathcal{A}$, we have:
		\[ g_a \sim \mathcal{N}(0, 1) \]
		\STATE Evaluate $Q[\mathbf{f} + \epsilon \mathbf{g}]$
		\STATE Estimate the directional derivative:
		\[ \nabla_{\mathbf{g}} Q = \frac{Q[\mathbf{f} + \epsilon \mathbf{g}] - Q[\mathbf{f}]}{\epsilon} \]
		\IF{$\nabla_{\mathbf{g}}Q > 0$}
			\STATE $\nabla Q \leftarrow \nabla Q + \frac{1}{k} \mathbf{g}$
		\ELSE
			\STATE $\nabla Q \leftarrow \nabla Q - \frac{1}{k} \mathbf{g}$
		\ENDIF
	\ENDFOR
\end{algorithmic}
\end{algorithm}
We are now ready to state our main result:
\begin{theorem}
Algorithm \ref{alg:estimateGradient} gives an unbiased estimate of the functional gradient of the functional $Q[f_\alpha]$.
\end{theorem}
\begin{proof} Since both the statistical model checking estimation and the gradient estimation are unbiased and independent of each other, this follows.
\end{proof}

\begin{algorithm}[ht!]
\caption{Stochastic gradient ascent for $Q[\mathbf{f}]$}
\label{alg:gradientAscent}
\begin{algorithmic}
\REQUIRE Initial function vector $\mathbf{f}_0$, learning rate $\gamma_0$, $n_{\max}$ iterations
\ENSURE A function vector $\mathbf{f}$ that approximates a local optimum of $Q$
	\FOR {$n \leftarrow 1$ \TO $n_{\max}$}
		\STATE Estimate the functional gradient $\nabla Q$ by using Algorithm \ref{alg:estimateGradient}
		\STATE Update: $\mathbf{f}_n \leftarrow \mathbf{f}_{n-1} + \gamma_{n-1} \nabla Q$
	\ENDFOR
\end{algorithmic}
\end{algorithm}

Therefore, we can use this stochastic estimate of the functional gradient to devise a stochastic gradient ascent algorithm which directly solves the variational problem in equation \eqref{varFormulation}.
This is summarised in Algorithm \ref{alg:gradientAscent}, which requires as input an initial vector of functions $\mathbf{f}_0$, and a learning rate $\gamma_0$.
The effects of the learning rate on the convergence properties of the method have been extensively studied in the literature.
In particular, for a decreasing learning rate convergence is guaranteed in the strictly convex scenario, if the following conditions are satisfied: $\sum_{n} \gamma_n = \infty$ and $\sum_{n} \gamma^2_n < \infty$ \cite{Murata1999,Bottou2010}, suggesting a $\Theta(n^{-1})$ decrease for the learning rate.
In non-convex problems, such as the ones considered in this work, the $\Theta(n^{-1})$ decrease is generally too aggressive, leading to vulnerability to local optima.
Following the recommendations of \cite{Bottou2012}, we adopt a more conservative strategy:
\begin{equation}
\gamma_n = \gamma_0\; n^{-1/2}
\end{equation}
where $\gamma_0$ is an initial value for the learning rate, which is problem dependent.

\section{Example}
\label{sec:example}
% !TEX root =  CTMDP_gradDescent.tex

We demonstrate the stochastic gradient ascent algorithm on a simple epidemiology that features no permanent recovery, also known as the SIS model.
The system is modelled as a PCTMDP, in which the state is described by two variables denoting the population of susceptible ($X_S$) and infected individuals ($X_I$).
We assume that no immunity to the infection is gained upon recovery.
The objective is to monitor how infection progresses over time, given that there is a non-deterministic choice  at each step among actions in $\mathcal{A} = \{\mathit{no\  treatment}, \mathit{treatment}\}$, indicating whether an external action is taken to deal with the infection.

This non-deterministic choice will affect the dynamics of the system, which are represented by a list of transitions together with their rate functions, in the biochemical notation style (see e.g. \cite{Gillespie1977}):
\begin{description}[labelwidth=150pt]
\item[infection (*):] $S+I \xrightarrow{k_i} I+I$,  with rate function  $k_i\, X_S\, X_I$;
\item[slow recovery (\emph{no treatment}):]  $I\xrightarrow{k_r} S$,  with rate function  $k_r\, X_I$;
\item[self-infection (\emph{no treatment}):] $S \xrightarrow{k_i} I$,  with rate function  $k_i\, X_S / 2$;
\item[fast recovery (\emph{treatment}):]     $I\xrightarrow{k_r} S$,  with rate function  $\alpha\, k_r\, X_I$;
\item[death (\emph{treatment}):]     $I\xrightarrow{k_r} \emptyset$,  with rate function  $k_d\, X_I$;
\item[death (\emph{treatment}):]     $S\xrightarrow{k_r} \emptyset$,  with rate function  $k_d\, X_S$;
\end{description}
Among the transitions above, only \emph{infection} has the same rate regardless of any non-deterministic choice.
If the \emph{no treatment} action is chosen, infected individuals recover slowly as prescribed by the \emph{slow recovery} transition, while there is a small chance of self-infection.
If treatment is applied, the recovery rate is increased by a factor $\alpha > 1$, and the chance of spontaneous infection is eliminated.
We assume however that the treatment is associated with some very negative side-effects that result in a small probability of death, either for healthy of infected individuals.

In this example, we seek to construct a scheduler that maximises the probability of having no deaths and no infected individuals during the time interval $[t_1,t_2]$, i.e. 
maximising the safety property 
\begin{equation}
	 \square_{[t_1,t_2]} G \qquad\qquad G=\{S = N\}
	% \eventually{40}{60} I> 0.1 N \wedge \eventually{0}{120} I = 0
\label{eq:extinction}
\end{equation}
%
%minimising the probability of reaching a state in $G = \{ S \neq N\}$. 
%This is equivalent to maximise the negation of the following time-bounded reachability property:
%\begin{equation}
%	\neg \diamond_{[t_1,t_2]} S \neq N
%	% \eventually{40}{60} I> 0.1 N \wedge \eventually{0}{120} I = 0
%\label{eq:extinction}
%\end{equation}
%where $N$ is total population in the initial state.
%The property in \eqref{eq:extinction} can be equivalently formulated as the safety property $\square_{[t_1,t_2]} S = N$, meaning that there are no infected individuals between times $t_1$ and $t_2$, and no death has occurred.
The application of treatment contributes in accelerating the extinction of the infected population, but it also introduces a possibility of death.
Therefore a policy of constantly applying treatment cannot be optimal with respect to the satisfiability of the property considered.
Moreover, maximising the satisfaction probability requires a time-dependent scheduler, as the treatment application has to be appropriately timed so that it has effect in the time-interval $[t_1, t_2]$.

In the experiments that follow, we illustrate how the stochastic gradient ascent algorithm converges to solutions that maximise this probability.
We consider a system with total population $N = 100$, and initial populations $X_{S_0} = 90$ and $X_{I_0} = 10$.
The rate constants are $k_i=0.0012$ for infection, $k_r=0.1$ for recovery, $k_d=0.0002$ for the death event, while the increase in the recovery rate due to treatment is fixed to $\alpha=10$.
The time bounds for the safety property considered are $t_1 = 50$ and $t_2 = 60$.
Regarding the stochastic gradient ascent parameters, the learning rate at the $n$-th step is $\gamma_n = \gamma_0 / \sqrt{n}$, where $\gamma_0 = 5$.
For the numerical estimation of the directional derivatives, we consider $\epsilon=0.1$ and the batch size for the gradient estimation was fixed to $k=5$.
For each estimation of the $Q$ function, we have used $1000$ simulation runs.
In all cases, the algorithm was run for $100$ iterations, meaning that a total of $600000$ simulation runs were used for each experiment.

We first present an example that illustrates the importance of time in the satisfaction of the time-bounded property in \eqref{eq:extinction}.
Figure \ref{fig:resulting_scheduler} reports a scheduler which is given as a solution by the stochastic gradient ascent approach.
The scheduler is presented as a multivariate function that takes values in $[0, 1]$, indicating the probability of selecting the \emph{no treatment} action for different values of state and time.
In particular, we have a series of surface plots, each of which summarises the probability of \emph{no treatment} as function of the 2-dimensional state-space for a different time-point.
The white colour denotes that \emph{no treatment} is selected with probability $1$, while the black colour implies that \emph{treatment} is used instead.
We can see that \emph{treatment} is only preferable for a particular time window and for certain parts of the state-space, that is $X_{S} > 80$ and $X_{I} < 20$. This makes sense, as the probability of achieving full recovery from a state with more than 20 infected is too small to justify the risks connected with treatment. 
More specifically, \emph{treatment} is selected with high probability for $t \in [33.75, 52.5]$, which precedes with a very small overlap the time interval if interest, which is $[50, 60]$.
Intuitively, to maximise the probability that all of the population is recovered over the course of a particular interval, the \emph{treatment} action should be engaged just before.
In a different case, there is an increased risk of death, as a consequence of the negative effects of prolonged treatment.

\begin{figure}[ht]
\centering
\subfigure {
	\includegraphics[width=0.3\textwidth]{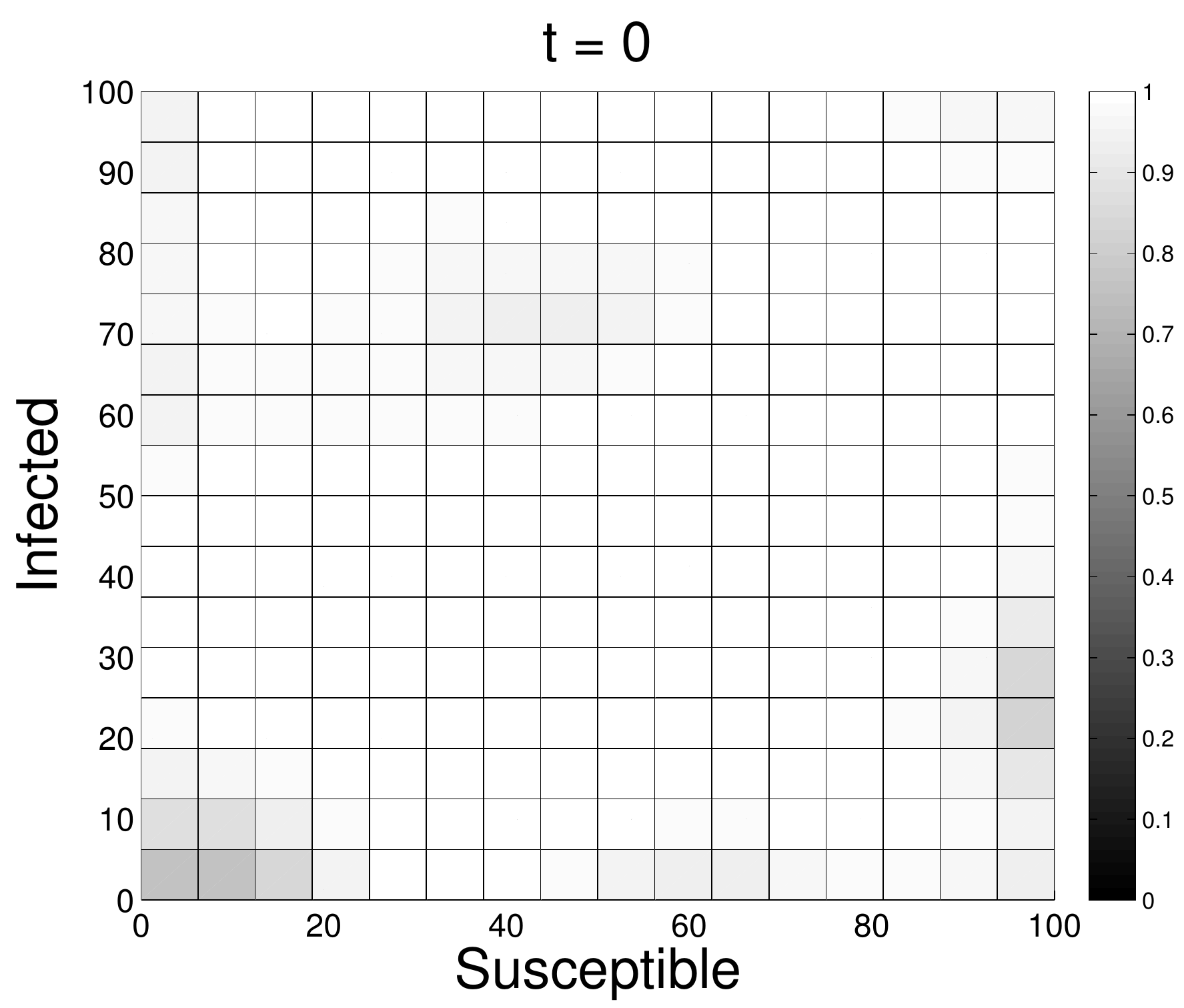}
}
\subfigure {
	\includegraphics[width=0.3\textwidth]{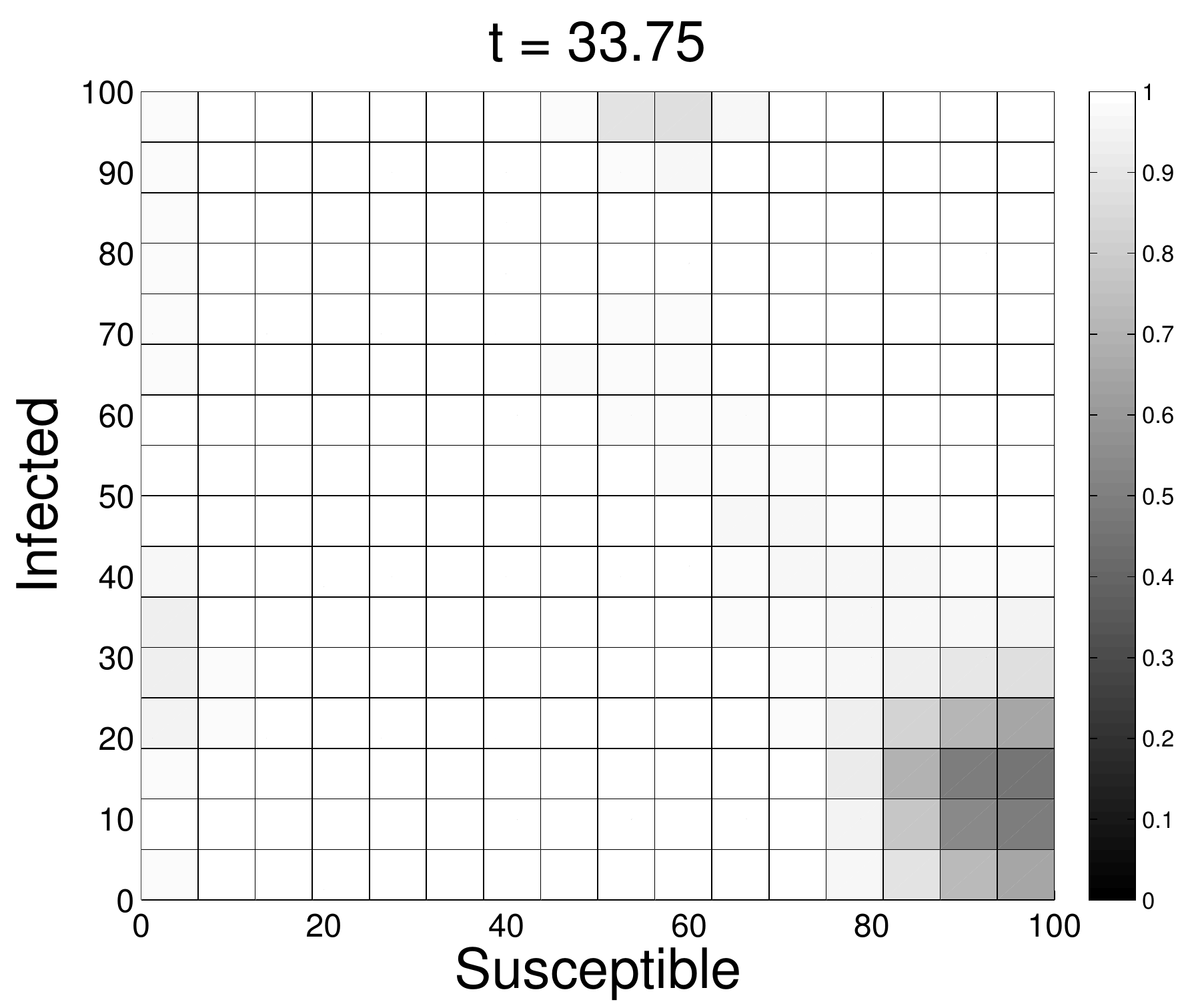}
}
\subfigure {
	\includegraphics[width=0.3\textwidth]{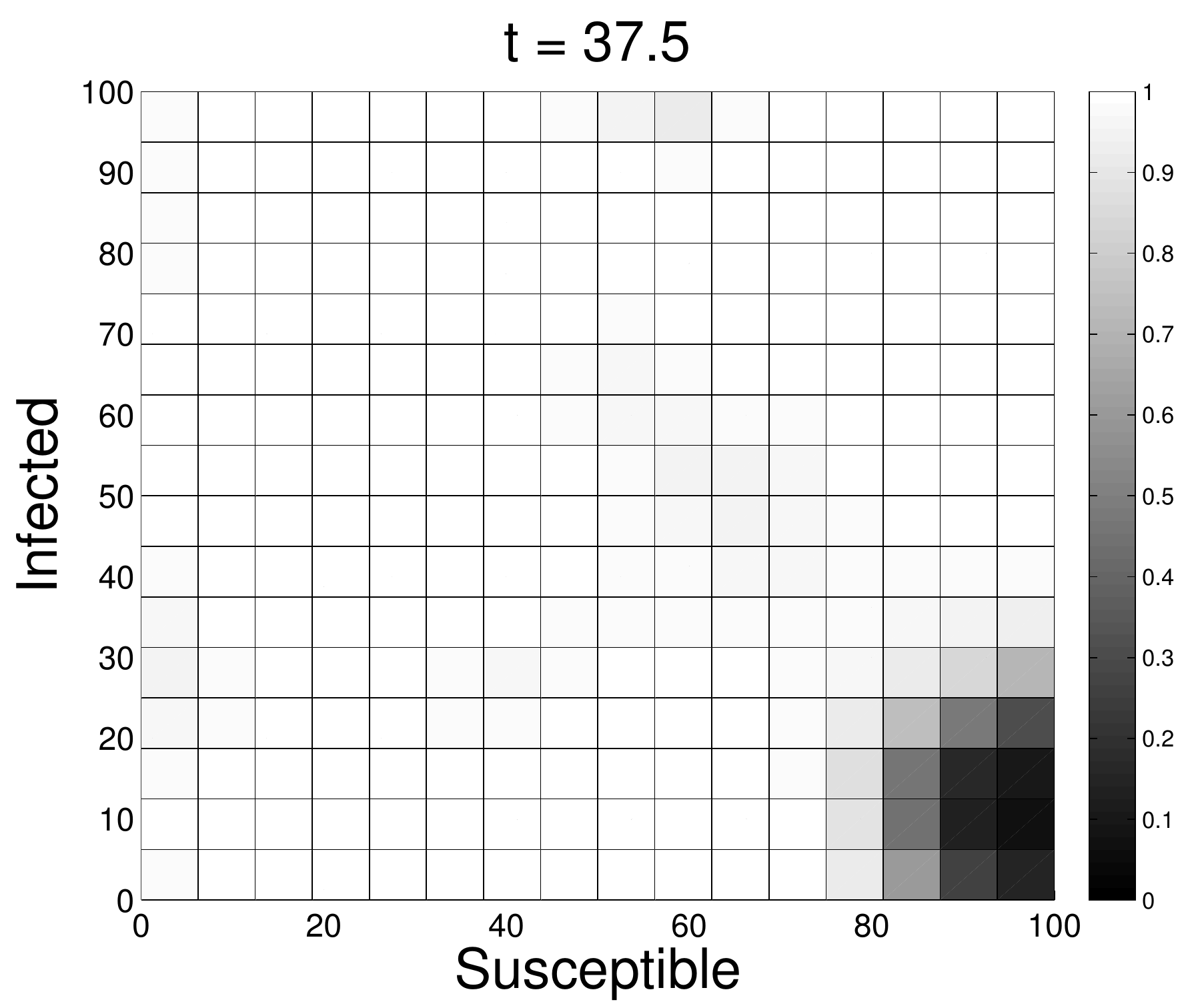}
}
\subfigure {
	\includegraphics[width=0.3\textwidth]{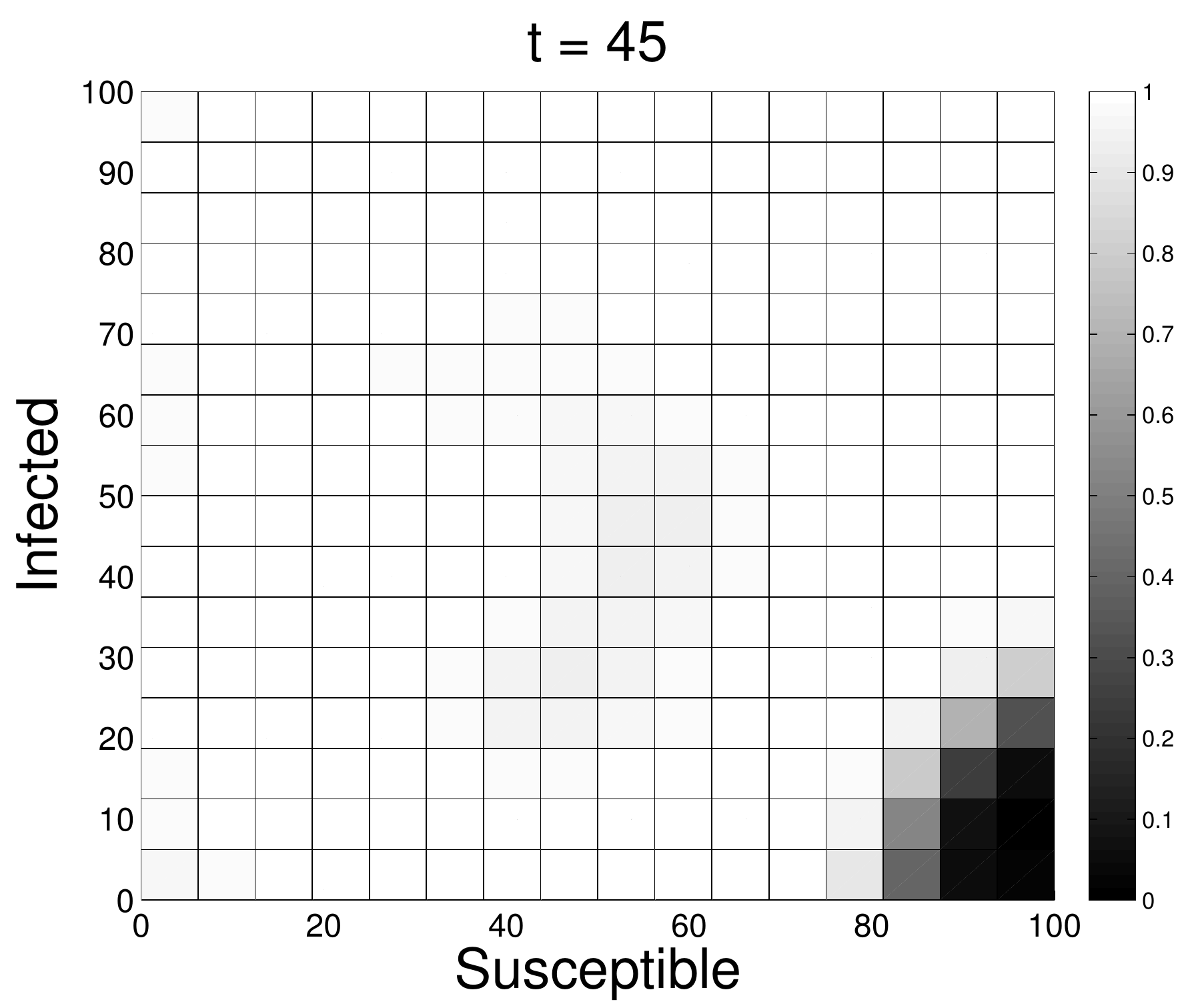}
}
\subfigure {
	\includegraphics[width=0.3\textwidth]{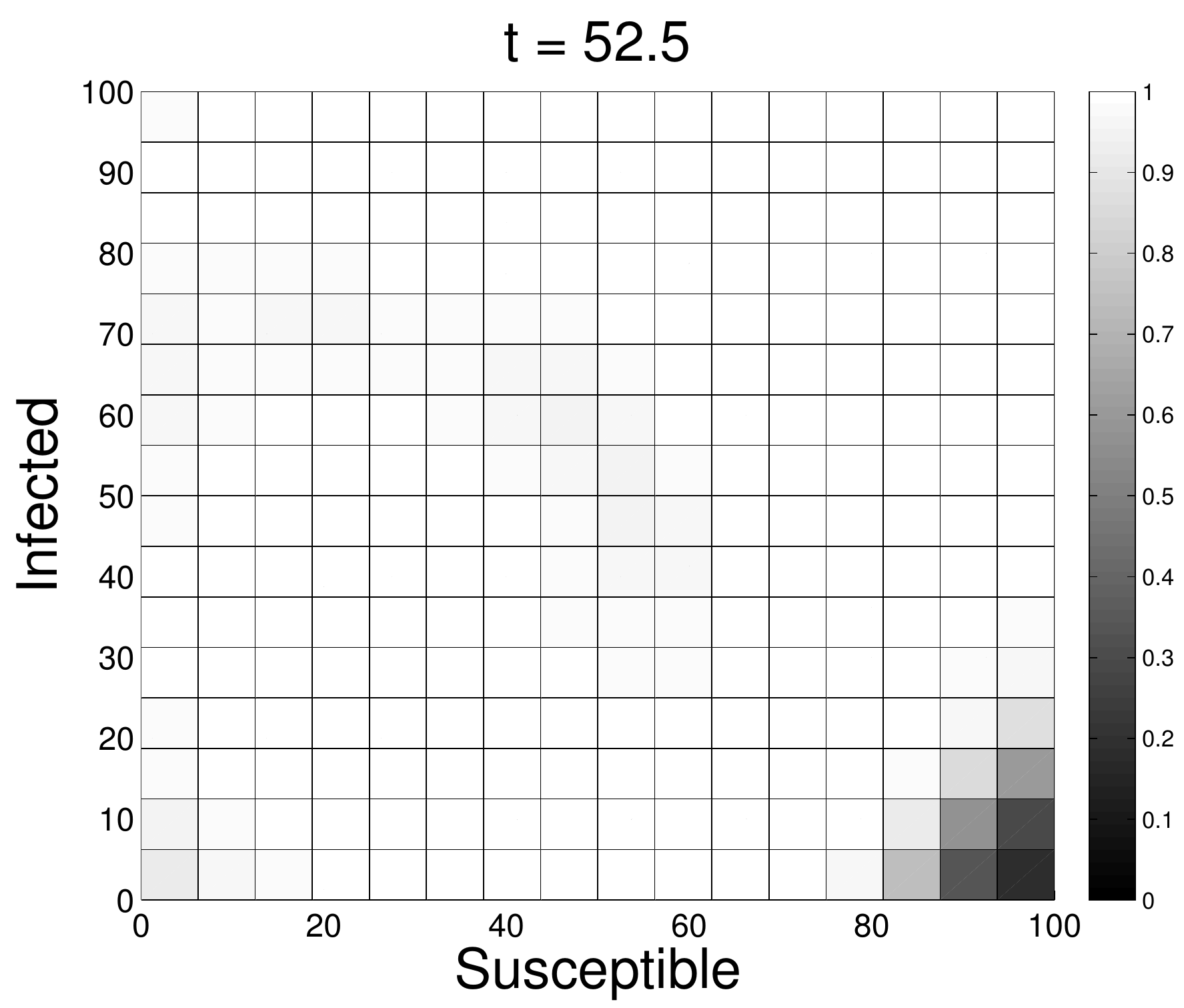}
}
\subfigure {
	\includegraphics[width=0.3\textwidth]{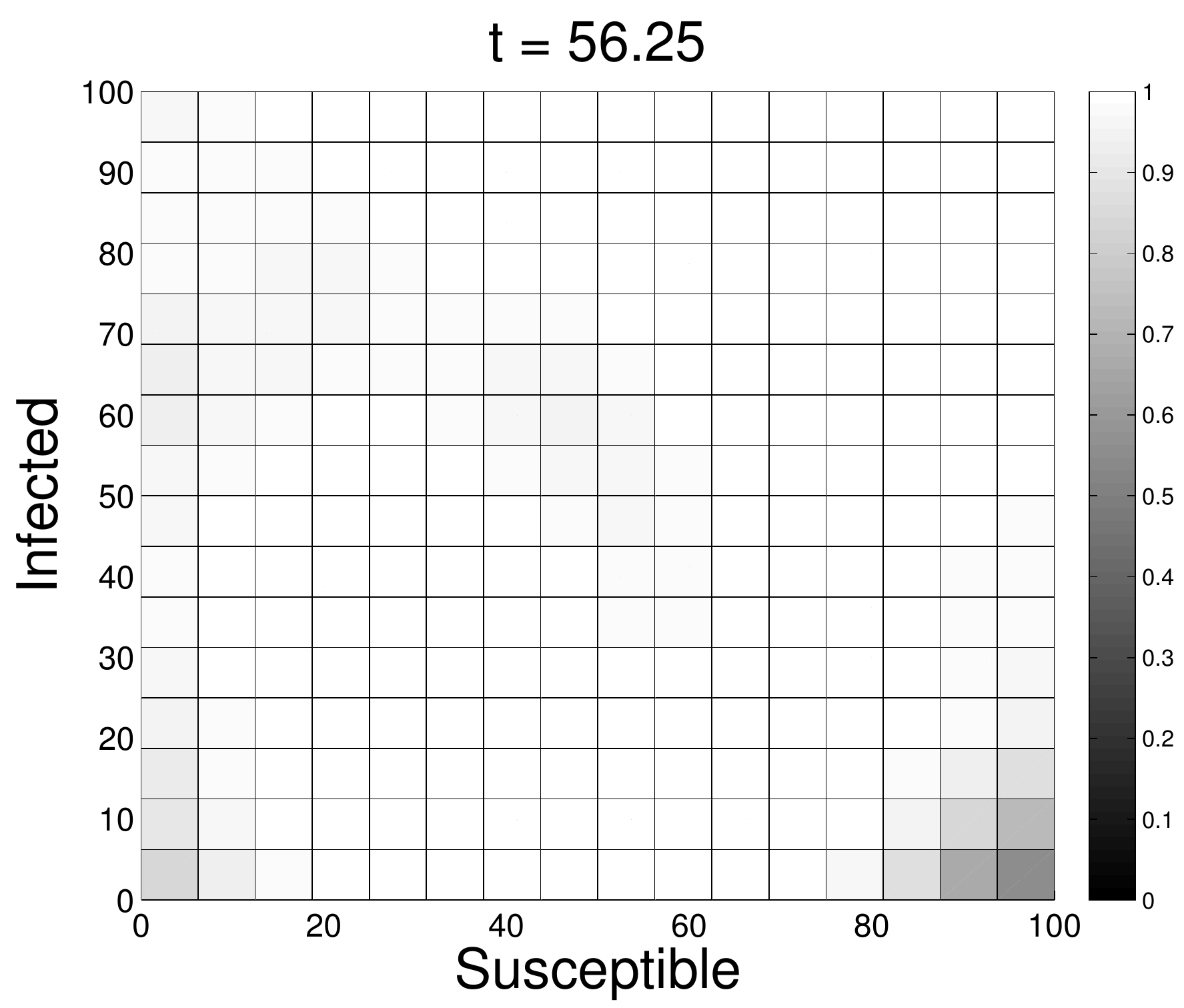}
}
\caption{Example of scheduler that (locally) maximises the probability of $\always{t_1}{t_2} S = N$. The white area indicates high probability of choosing the \emph{no treatment} action; the dark area indicates high probability of choosing \emph{treatment}.}
\label{fig:resulting_scheduler}
\end{figure}

We next investigate how the algorithm responds to different initial schedulers.
In Figure \ref{fig:qvalues}, we monitor how the value of the functional $Q$ as function of the scheduler evolves during the course of the algorithm, starting from different initial solutions.
More specifically, Figure \ref{fig:qvalues_1} depicts the evolution of $Q$ values starting from a scheduler where \emph{no treatment} is globally selected as an action.
The initial satisfaction probability is very small, but after a number of iterations it converges to values above $0.6$.
Figure \ref{fig:qvalues_2} summarises the results where the initial solution selects \emph{treatment} everywhere; apparently this initial solution has been closer to the local optimum and the convergence rate had been significantly faster in this case.
Convergence is even faster in Figure \ref{fig:qvalues_3}, where a uniform initial solution was used; that is that each of the two possible actions has equal probability $\forall s \in S$ and $\forall t \in T$.
Finally, in Figure \ref{fig:qvalues_4} we report the $Q$ values for a run starting from a randomly initialised scheduler.
In the last two instances, the starting point has had $Q$ values at around $0.4$, which is closer to the maximum; therefore the algorithm naturally required fewer iterations to converge to a good solution.
Although the convergence rate is apparently dependent on the initial solution, the experiments considered resulted in solutions of similar value, which obtain satisfaction probabilities at around $0.65$.
It is important to note however that there is no guarantee that the algorithm will converge to the global maximum, since the problem considered in not convex in the general case.

\begin{figure}[ht]
\centering
\subfigure [\emph{no treatment} only initial scheduler \label{fig:qvalues_1}] {
	\includegraphics[width=0.44\textwidth]{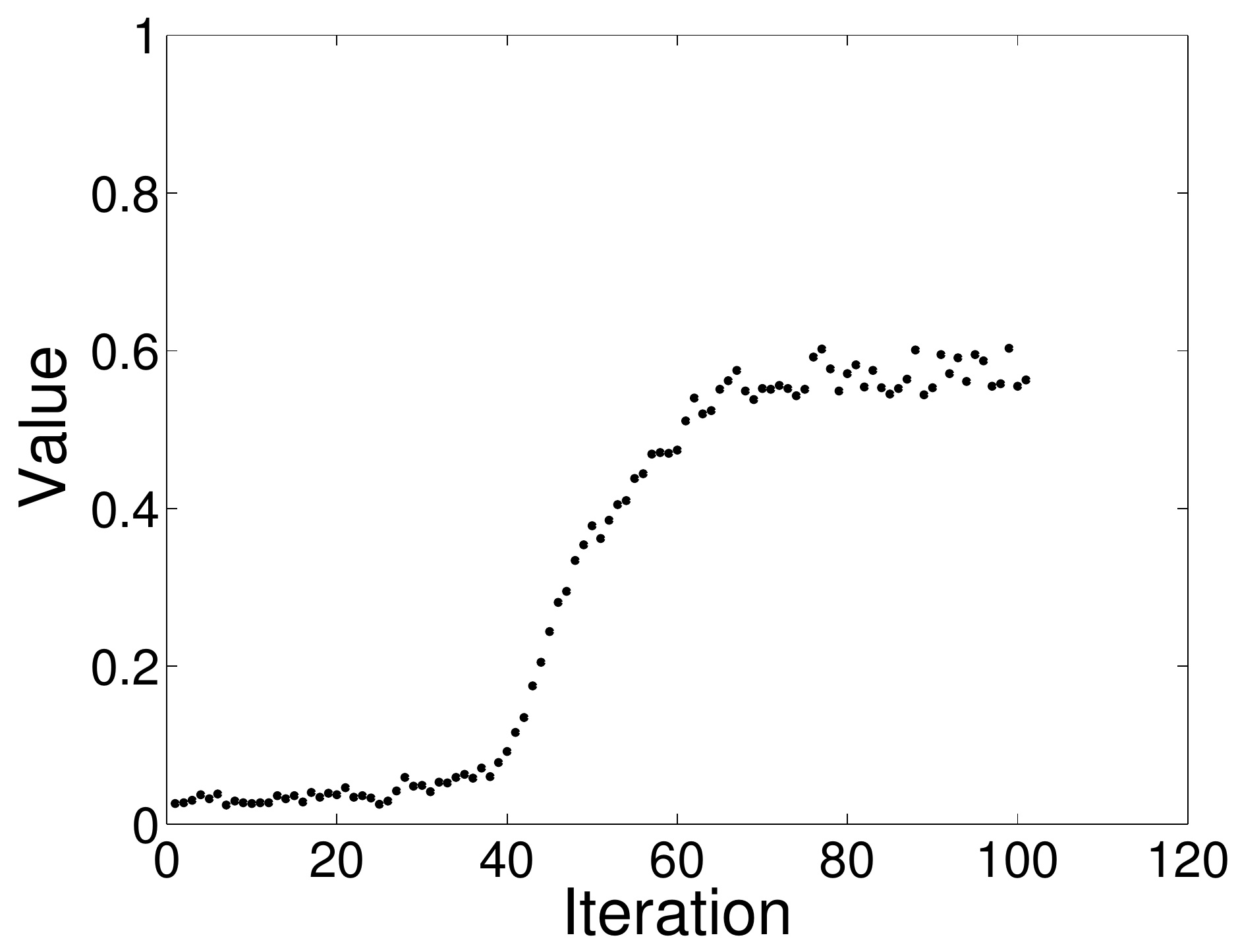}
}
\subfigure [\emph{treatment} only initial scheduler \label{fig:qvalues_2}] {
	\includegraphics[width=0.44\textwidth]{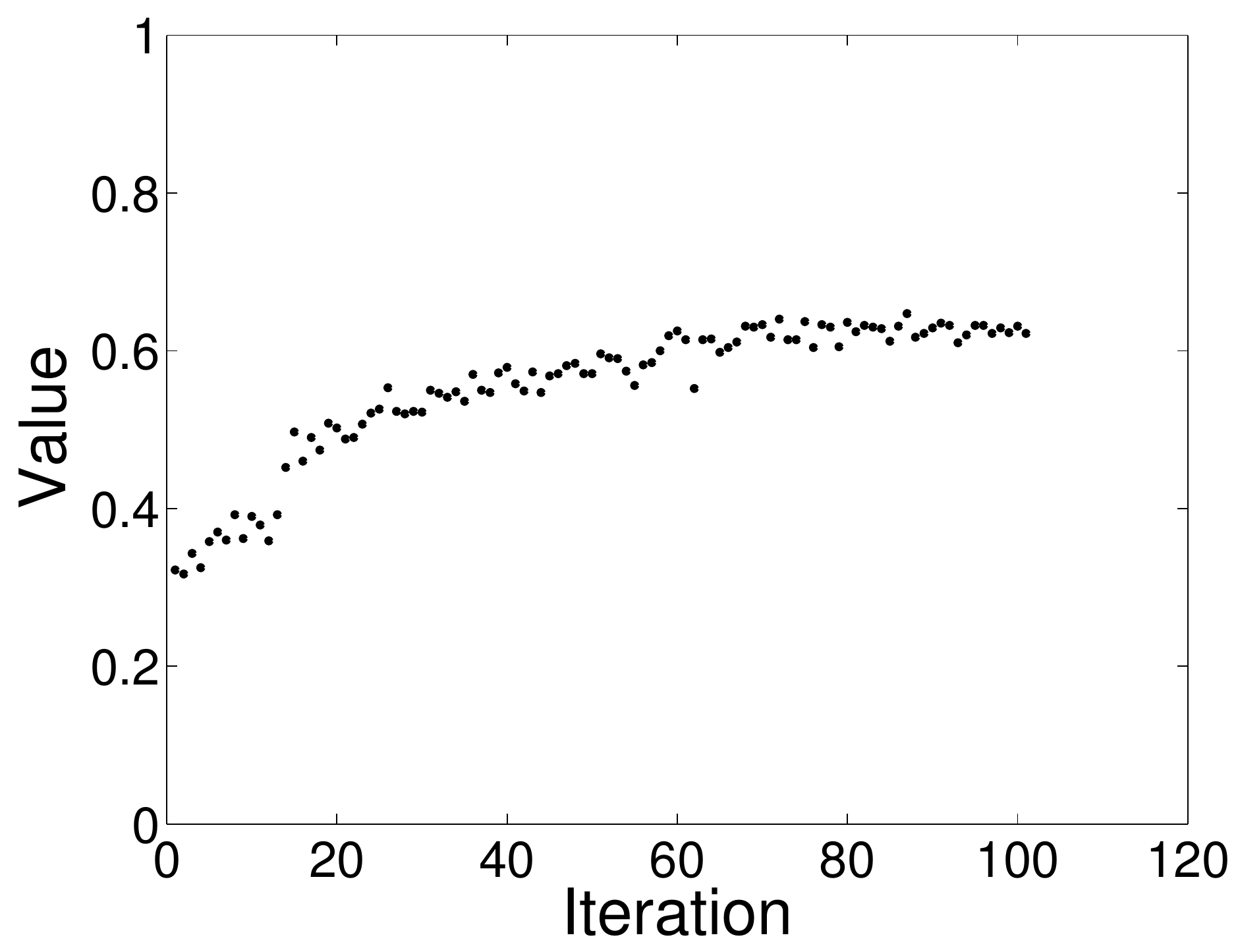}
}
\subfigure [Uniform initial scheduler \label{fig:qvalues_3}] {
	\includegraphics[width=0.44\textwidth]{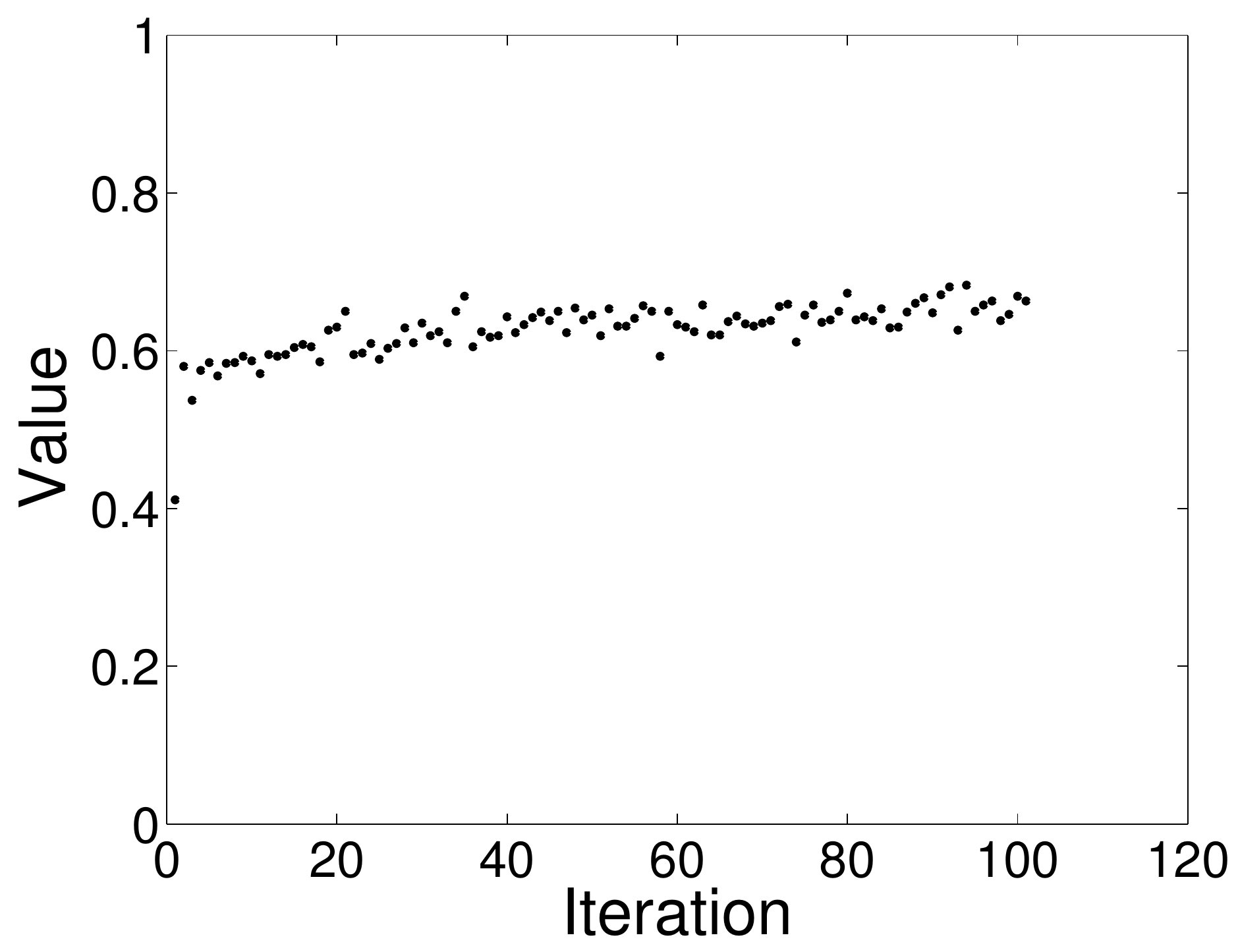}
}
\subfigure [Random initial scheduler \label{fig:qvalues_4}] {
	\includegraphics[width=0.44\textwidth]{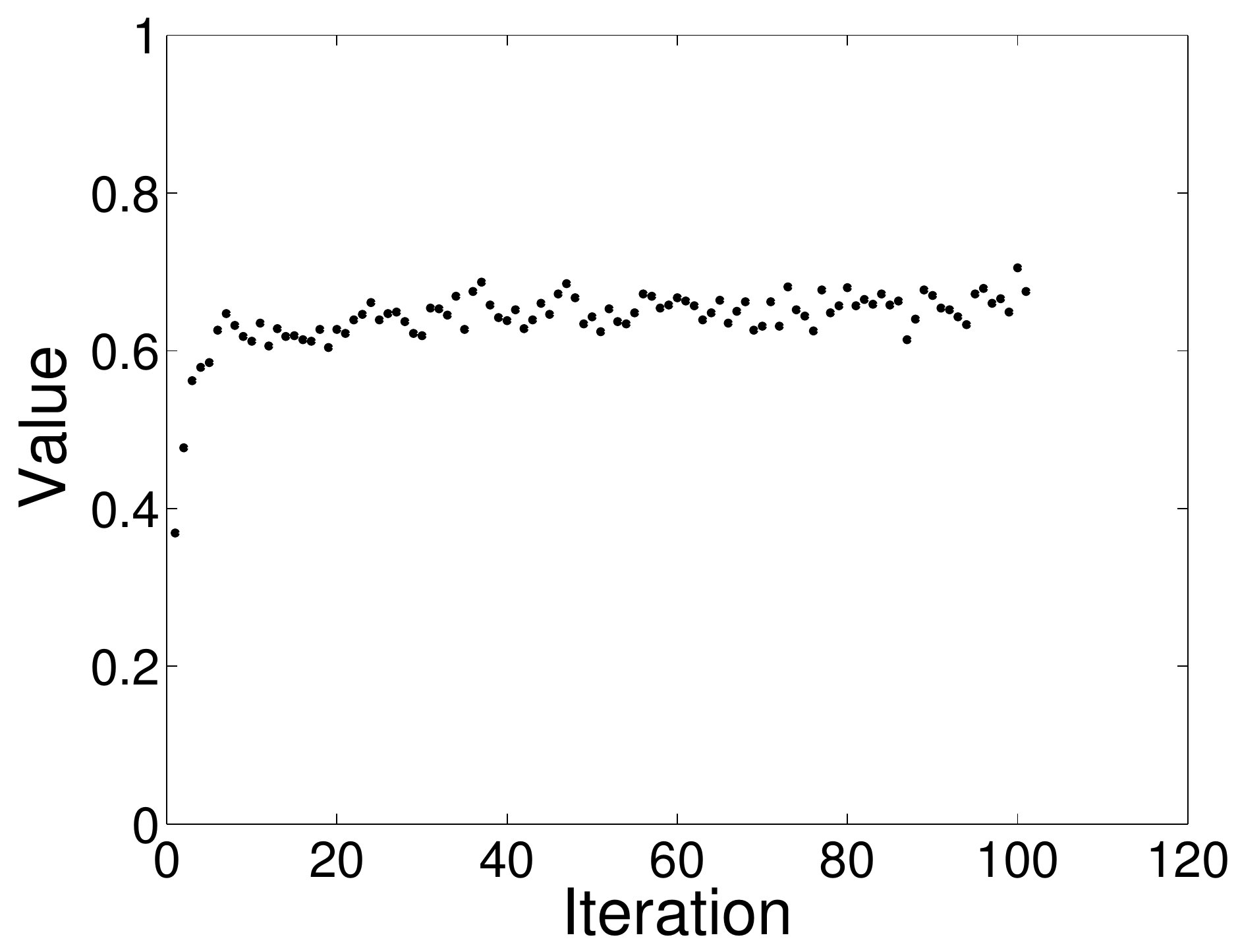}
}
\caption{Stochastic gradient ascent starting from different initial schedulers}
\label{fig:qvalues}
\end{figure}

\section{Conclusions}
\label{sec:conclusions}
% !TEX root =  CTMDP_gradDescent.tex
 Continuous time Markov Decision processes play an important role in many applications, yet they are relatively understudied in the formal methods literature. Part of the problem resides in the difficulty to provide effective characterisations of time-varying schedulers. Recent methodologies \cite{Yuliya2015} have focussed on iterative algorithms based on uniformisation over an increasingly fine time discretisation. While such methods have the ability to compute exactly (up to numerical precision) the objective function (reachability probability), their scalability to large systems is significantly hampered by the state-space explosion problem. Furthermore, such approaches rely on the availability of a mathematical description of the systems, and are therefore not applicable to control black-box systems where a reliable model is not available.

 Our approach is suitable instead when the model of the system we want to control is not available a-priori. Our algorithm relies on using GPs, a probability distribution over the space of functions which universally approximates continuous functions. 
 
 A potentially significant limitation of our approach is its vulnerability to locally optimal choices. This is a common problem in optimisation, where global convergence in the non-convex case is well known to be hard. Theoretically, this means that our approach can only provide a lower-bound on the reachability probability; nevertheless, this can still be a very valuable result in practical scenarios. Empirically, we observed that the algorithm had excellent performance in a challenging test set; its computational efficiency also means that practical strategies to avoid local optima, such as multiple restarts, can be feasibly employed.

\paragraph{Acknowledgements.}
L.B.\ acknowledges partial support from the EU-FET project QUANTICOL (nr. 600708) and by FRA-UniTS.
G.S.\ and D.M.\ acknowledge support from the European Reasearch Council under grant MLCS306999.
T.B.\ is supported by the Czech Science Foundation, grant No.~\mbox{15-17564S}.
E.B.\ acknowledges the partial support of the Austrian National 
Research Network  S 11405-N23 (RiSE/SHiNE) of the Austrian Science Fund (FWF),
the ICT COST Action IC1402 Runtime Verification beyond Monitoring (ARVI) and the IKT der Zukunft of Austrian FFG project HARMONIA (nr. 845631).

\bibliographystyle{abbrv}
\bibliography{biblio}

\appendix
\section{Appendix}
\subsection{Proof of Proposition~\ref{prop:Neuhausser}}
In general, a {\em history-dependent randomized (HR)} scheduler $\pi$ is a (measurable) function which takes a path (a history) $h=s_0 t_0 s_1 t_1 \cdots s_n$ and returns a probability distribution on actions of $\mathcal{A}$. We write $\pi(h, a)$ to denote the probability that $a$ is taken after the history $h$. Our schedulers, as defined in~Defintion~\ref{def:scheduler}, are called {\em total time-positional randomized (TTPR)} schedulers. If the scheduler always assigns the probability one to exactly one action, we say that it is {\em deterministic}, which gives us classes HD and TTPD of history-dependent deterministic and total time-positional deterministic schedulers.
In principle, it has been shown in~\cite{Neuhausser2010} that our restriction is without loss of generality. We include a sketch of the argument {\em just for completeness}.

The argument can be (roughly) summarized as follows: Let us add a counter to 
the state-space i.e., states are now of the form $(s,k)$ where $s$ is a state of the original CTMDP $\mathcal{M}$ and $k$ is the number of steps the process made from the beginning. The CTMDP $\mathcal{M}$ is simulated in the first component and the number of steps counted in the other one, up to the moment when a threshold $n+1$ is reached and from this moment on the counter stays at value $n+1$ forever. The new goal states are the pairs 
$(s,k)$ where $s$ is a goal state in $\mathcal{M}$ and $k\leq n$. 
This gives us a new CTMDP $\mathcal{M}_n$.
Note that every HR scheduler in $\mathcal{M}_n$ can be easily transformed into a HR scheduler in $\mathcal{M}$ by taking a projection on the first component.

Denote by $V^n((s,k),t)$ the probability of reaching a goal state in $\mathcal{M}_n$ from $(s,k)$ within the time interval $I-t=[\max(0,t_1-t),\max(0,t_2-t)]$ where $I=[t_1,t_2]$.
Values $V^n((s,k),t)$ in the CTMDP $\mathcal{M}_n$ can be computed using backward induction as follows: Clearly, $V^n((s,n+1),t)$ is $0$ for all $t$. 
Assume that we already have $V^n((s,k+1),t)$. Now it suffices to find $\pi^n$ so that the following is maximized:
\[
\sum_{a} \pi^n(t,(s,k),a)\sum_{s'}\int_{0}^{\infty} R(s,a,s')e^{-R(s,a,s')t'} V^n((s',k+1),t
+t')dt'
\]
(Intuitively, first $a$ is chosen with probability $\pi^n(t,(s,k),a)$, then time delay $t'$ is chosen from the exponential distribution together with the next state $(s',k+1)$, finally we proceed optimally from $(s',k+1)$ after time $t+t'$, which means that we reach a goal state with probability $V^n((s',k+1), t+t')$.)
Apparently, it is optimal to choose
\[
\pi^n(t,(s,k),a) \in \mathit{argmax}_{a} \sum_{s'}\int_{0}^{\infty} R(s,a,s')e^{-R(s,a,s')t'} V((s',k+1),t+t')dt'
\]
Now observe that for every $k$ and every $t$ we have $\lim_{n\rightarrow\infty} V^n(s,k,t)=V(s,t)$ where $V(s,t)=\sup_{\sigma\in\Sigma} \mathbb{P}^{\mathcal{M},s}_{\sigma}(\diamond_{I-t} G)$.

Now let $m$ be large enough so that the probability of making more than $m$ steps in at most $t_2$ time units is less than $\varepsilon$. 
It follows that the strategy $\pi^{2m}$, which is optimal in $\mathcal{M}^{2m}$, is $\varepsilon$-optimal in $\mathcal{M}$ (which means that it satisfies $\diamond_{I} G$ with probability $\varepsilon$-close to the maximum value).

Let $m'>2m$ be large enough so that for all $k\leq 2m$ and all $t\leq t_2$ we have that 
\begin{align*}
\mathit{argmax}_{a} & \sum_{s'}\int_{0}^{\infty} R(s,a,s')e^{-R(s,a,s')t'} V^{m'}(s',k+1,t+t')dt' = \\ & \mathit{argmax}_{a}
\sum_{s'} \int_{0}^{\infty} R(s,a,s')e^{-R(s,a,s')t'} V(s',t+t')dt'
\end{align*}
It follows that a strategy which always chooses an action from 
\[
\mathit{argmax}_{a}\sum_{s'} \int_{0}^{\infty} R(s,a,s')e^{-R(s,a,s')t'} V(s',t+t')dt'
\]
behaves similarly to $\pi^{m'}$ and hence is $\varepsilon$-optimal. As $\varepsilon>0$ was chosen arbitrarily and the above choice depends only on $s$ and $t$, we obtain the desired optimal TTPD scheduler.\qed

\end{document}